\documentclass[prl,nopacs,twocolumn,twoside,superscriptaddress]{revtex4}

\usepackage{amsmath}
\usepackage{amsfonts,amssymb,graphicx,latexsym,color,revsymb}

\def\squareforqed{\hbox{\rlap{$\sqcap$}$\sqcup$}}
\def\qed{\ifmmode\squareforqed\else{\unskip\nobreak\hfil
\penalty50\hskip1em\null\nobreak\hfil\squareforqed
\parfillskip=0pt\finalhyphendemerits=0\endgraf}\fi}
\def\endenv{\ifmmode\;\else{\unskip\nobreak\hfil
\penalty50\hskip1em\null\nobreak\hfil\;
\parfillskip=0pt\finalhyphendemerits=0\endgraf}\fi}

\newtheorem{theorem}{Theorem}

\newtheorem{lemma}[theorem]{Lemma}

\newtheorem{proposition}[theorem]{Proposition}
\newtheorem{remark}[theorem]{Remark}

\newlength{\blank}
\newenvironment{proof}[1][{\hspace{-\blank}}]{{\noindent\textbf{Proof~{#1}.\ }}}{\hfill\qed}
\newcommand{\ket}[1]{|#1\rangle}
\newcommand{\bra}[1]{\langle#1|}

\mathchardef\ordinarycolon\mathcode`\:
\mathcode`\:=\string"8000
\def\vcentcolon{\mathrel{\mathop\ordinarycolon}}
\begingroup \catcode`\:=\active
  \lowercase{\endgroup
  \let :\vcentcolon
  }

\newcommand{\nc}{\newcommand}
\nc{\rnc}{\renewcommand}
\nc{\beq}{\begin{equation}}
\nc{\eeq}{{\end{equation}}}
\nc{\beqa}{\begin{eqnarray}}
\nc{\eeqa}{\end{eqnarray}}
\nc{\lbar}[1]{\overline{#1}}
\nc{\ketbra}[2]{|#1\rangle\!\langle#2|}
\nc{\proj}[1]{| #1\rangle\!\langle #1 |}
\nc{\avg}[1]{\langle#1\rangle}
\rnc{\max}{\operatorname{max}}
\nc{\Rank}{\operatorname{Rank}}
\nc{\smfrac}[2]{\mbox{$\frac{#1}{#2}$}}
\nc{\tr}{\operatorname{Tr}}
\nc{\ox}{\otimes}
\nc{\dg}{\dagger}
\nc{\dn}{\downarrow}
\nc{\cA}{\mathcal{A}}
\nc{\cB}{\mathcal{B}}
\nc{\cC}{\mathcal{C}}
\nc{\cD}{\mathcal{D}}
\nc{\cE}{\mathcal{E}}
\nc{\cF}{\mathcal{F}}
\nc{\cG}{\mathcal{G}}
\nc{\cH}{\mathcal{H}}
\nc{\cI}{\mathcal{I}}
\nc{\cJ}{\mathcal{J}}
\nc{\cK}{\mathcal{K}}
\nc{\cL}{\mathcal{L}}
\nc{\cM}{\mathcal{M}}
\nc{\cN}{\mathcal{N}}
\nc{\cO}{\mathcal{O}}
\nc{\cP}{\mathcal{P}}
\nc{\cR}{\mathcal{R}}
\nc{\cS}{\mathcal{S}}
\nc{\cT}{\mathcal{T}}
\nc{\cX}{\mathcal{X}}
\nc{\cZ}{\mathcal{Z}}
\nc{\csupp}{{\operatorname{csupp}}}
\nc{\qsupp}{{\operatorname{qsupp}}}
\nc{\var}{\operatorname{var}}
\nc{\rar}{\rightarrow}
\nc{\lrar}{\longrightarrow}
\nc{\polylog}{\operatorname{polylog}}
\nc{\1}{{\openone}}
\nc{\id}{{\operatorname{id}}}

\nc{\RR}{{{\mathbb R}}}
\nc{\CC}{{{\mathbb C}}}
\nc{\FF}{{{\mathbb F}}}
\nc{\NN}{{{\mathbb N}}}
\nc{\ZZ}{{{\mathbb Z}}}
\nc{\PP}{{{\mathbb P}}}
\nc{\QQ}{{{\mathbb Q}}}
\nc{\UU}{{{\mathbb U}}}
\nc{\EE}{{{\mathbb E}}}

\nc{\be}{\begin{equation}}
\nc{\ee}{{\end{equation}}}
\nc{\bea}{\begin{eqnarray}}
\nc{\eea}{\end{eqnarray}}
\nc{\<}{\langle}
\nc{\Hom}[2]{\mbox{Hom}(\CC^{#1},\CC^{#2})}
\nc{\rU}{\mbox{U}}

\begin{document}

\title{The private capacity of quantum channels is not additive}

\author{Ke Li}
\email{leeke@mail.ustc.edu.cn}
\affiliation{Key Laboratory of Quantum Information, University of
Science and Technology of China, Chinese Academy of Sciences, Hefei,
Anhui 230026, China}

\author{Andreas Winter}
\email{a.j.winter@bris.ac.uk} \affiliation{Department of
Mathematics, University of Bristol, University Walk, Bristol BS8
1TW, U.K.} \affiliation{Centre for Quantum Technologies, National
University of Singapore, 2 Science Drive 3, Singapore 117542}

\author{XuBo Zou}
\email{xbz@ustc.edu.cn} \affiliation{Key Laboratory of Quantum
Information, University of Science and Technology of China, Chinese
Academy of Sciences, Hefei, Anhui 230026, China}

\author{GuangCan Guo}
\email{gcg@ustc.edu.cn}
\affiliation{Key Laboratory of Quantum Information, University of Science and Technology of China, Chinese Academy of Sciences, Hefei, Anhui 230026, China}

\date{30 August 2009}

\begin{abstract}
Recently there has been considerable activity on the subject
of additivity of various quantum channel capacities.
Here, we construct a family of channels with sharply bounded
classical, hence private capacity.
On the other hand, their quantum capacity when
combined with a zero private (and zero quantum) capacity erasure
channel, becomes larger than the previous classical capacity.
As a consequence, we can conclude for the first time
that the classical private capacity is
non-additive. In fact, in our construction even the quantum capacity of
the tensor product of two channels can be greater than the sum
of their individual classical private capacities.
We show that this violation occurs quite generically: every channel
can be embedded into our construction, and a violation occurs whenever
the given channel has larger entanglement assisted quantum capacity
than (unassisted) classical capacity.
\end{abstract}

\maketitle

Information Theory, established by Claude Shannon in the
1940s as a ``Mathematical Theory of
Communication''~\cite{Shannon48}, is the theoretical foundation of
today's communication technologies. The main problem it is
concerned with is how much information can be transmitted
down a noisy channel asymptotically, i.e.~the capacity of the
channel. Shannon provided a beautifully simple formula for the
capacity of a discrete memoryless channel, which only involves an
entropic expression of a single channel use. Subsequent research
revealed that this simple capacity formula fully characterizes the
information-carrying capability of a channel under a large range of
circumstances~\cite{CT91}, serving as a very robust measure. E.g.,
the ability of two channels together to transmit
information is quantified by the sum of their individual capacities.

Our world however is not the classical one of Shannon's noisy
channels, but is at a basic level described by quantum theory. To
understand the ultimate limit the laws of physics impose on our
ability to communicate, the underlying quantum behavior
of the channels should be considered. A quantum channel $\cN$ is
mathematically described by an isometric map $V$ from the input Hilbert
space $A$ to the combined Hilbert space of the output $B$ and the so-called
environment system $E$. Then the channel and its natural complement
$\widetilde{\cN}$ act as
\[
  \cN(\rho) = \tr_E V\rho V^\dagger,\quad
  \widetilde{\cN}(\rho) = \tr_B V\rho V^\dagger.
\]
It can in general not only convey classical
messages, but also quantum data, i.e.~a Hilbert space of quantum
states. It can also carry classical private information,
inaccessible to the environment, enabling the classically
impossible, provably unconditionally secure key
distribution~\cite{BB84}. Naturally, deriving capacity formulae of a
quantum channel for transmitting various kinds of information is a
central task of quantum information theory.

The classical capacity, $C(\cN)$, is the maximal rate of classical
information that the quantum channel $\cN$ can asymptotically
transmit with vanishing errors.
%
%
In contrast to the classical capacity, the definition of classical
private capacity $P(\cN)$ further requires that the
classical information conveyed is secret from the environment.
%
Finally, the quantum capacity $Q(\cN)$ quantifies how large a
Hilbert space of states the channel $\cN$ can transmit
asymptotically and with the error approaching zero. Operationally,
quantum information transmission implies private classical
transmission, which in turn implies plain classical
communication. I.e.,
\begin{equation}
  \label{eq:Capacities-Relation}
  C(\cN) \geq P(\cN) \geq Q(\cN).
\end{equation}

Despite considerable progress, tractable formulae for the quantum,
private and unrestricted classical capacities are still out of
reach. The HSW theorem~\cite{HSW9798}, Devetak~\cite{Devetak03} and
the LSD theorem~\cite{Lloyd97, Shor02, Devetak03} give the
classical, private and quantum capacities, respectively, as the
regularisation of single-letter quantities:
\begin{align}
  \label{eq:C-capacity}
  \chi(\cN)    &\leq C(\cN) = \lim_{n \rightarrow \infty}\frac{1}{n}\chi(\cN^{\otimes{n}}), \\
  \label{eq:P-capacity}
  P^{(1)}(\cN) &\leq P(\cN) = \lim_{n \rightarrow \infty}\frac{1}{n}P^{(1)}(\cN^{\otimes{n}}), \\
  \label{eq:Q-capacity}
  Q^{(1)}(\cN) &\leq Q(\cN) = \lim_{n \rightarrow \infty}\frac{1}{n}Q^{(1)}(\cN^{\otimes{n}}).
\end{align}
All three single-letter quantities are obtained via finite optimizations of
entropic expressions: the \emph{Holevo capacity} $\chi(\cN)$ is the maximum over all
ensembles $\{p_i,\rho_i\}$ of states on $A$ of the \emph{Holevo information}
\begin{equation}
  \label{eq:H-information}
  \chi_{\{p_i,\rho_i\}}(\cN) = H\!\left(\cN\!\left(\sum_i p_i \rho_i\right)\!\right)\!
                                 -\sum_i p_i H\bigl(\cN(\rho_i)\bigr),
\end{equation}
where $H(\rho)= -\tr \rho \log \rho$ is the von Neumann entropy
($\log$ is always the binary logarithm).
Similarly,
$P^{(1)}(\cN)= \max_{\{p_i,\rho_i\}} \Bigl(\chi_{\{p_i,\rho_i\}}(\cN)
                                            -\chi_{\{p_i,\rho_i\}}(\widetilde{\cN})\Bigr)$,
and $Q^{(1)}(\cN)=\max_{\rho}I_c(\rho,\cN)$,
with the \emph{coherent information}~\cite{Schumacher:I-coh}
\begin{equation}
  \label{eq:Co-information}
  I_c(\rho,\cN)=H(\cN(\rho))-H(\widetilde{\cN}(\rho)).
\end{equation}

Neither $\chi(\cN)$, nor $Q^{(1)}(\cN)$, nor $P^{(1)}(\cN)$ are additive;
in fact, $\chi \neq C$~\cite{Hastings08}, $P^{(1)} \neq P$~\cite{SmithRenesSmolin-P1},
$Q^{(1)} \neq Q$~\cite{ShorSmolin-Q1}. However, for certain classes
of channels it is known that $C(\cN)=\chi(\cN)$~\cite{King02,King03,Shor01},
and for other classes $P(\cN)=P^{(1)}(\cN)$, $Q(\cN)=Q^{(1)}(\cN)$~\cite{DevetakShor05}.

As measures of a channel's information transmitting capability, the
above three capacity quantities might be expected to be robust,
i.e.~just like Shannon's capacity for classical channels, to be
applicable under a large range of settings. While this is no longer
true when various auxiliary resources (e.g.~free entanglement or
classical communications) are available~\cite{BDSS04}, another weird
feature of the quantum capacity was discovered recently. Smith and
Yard~\cite{SY08} show that, as a function of channels, $Q(\cN)$ is
not additive. Specifically, for the two channels $\cN_1$ and $\cN_2$
with $\cN_1$ satisfying $Q(\cN_1)=0$ and $P(\cN_1)>0$, and $\cN_2$
the (zero quantum and zero private capacity) $50\%$ erasure channel,
$Q(\cN_1 \otimes \cN_2)\geq \frac{1}{2} P(\cN_1)>0$. One
might attribute this superactivation of quantum capacity to the
ability to transmit privacy \cite{Oppenheim08}, recalling the close
relationship between $Q(\cN)$ and $P(\cN)$~\cite{DW03}.
But surprisingly again, Smith and Smolin~\cite{SS08} have found
two channels such that either they have large
joint quantum capacity but negligible individual private classical
capacities, or one of them exhibits a large non-additivity of $\chi$.

In this Letter, we present quantum channels $\cT^k_{\cN}$
for given channel $\cN$ with finite environment dimension
(this includes all channels with finite dimensional input and output), and
integer $k$; it inherits input and output from $\cN$, but has also
auxiliary registers.
We can show that $C(\cN) \leq C(\cT^k_{\cN}) \leq C(\cN) + \delta(k)$,
where $\delta(k)$ goes to zero as $k\rar\infty$. Regarding the
capability of the channel $\cT^k_{\cN}$, together with a 50\%{}
erasure channel $\cA$, for
quantum communication, we find that the quantum capacity of the
combined channel $\cT^k_{\cN}\otimes \cA$ is lower bounded by
$Q_E(\cN)$, the entanglement-assisted quantum capacity of
$\cN$~\cite{BSST02}. So, for channels $\cN$ such that
$Q_E(\cN)>C(\cN)$, $\cT^k_{\cN}$ -- when combined with the above
erasure channel -- can transmit more quantum information than its
classical capacity $C(\cT^k_{\cN})$. Referring to
Eq.~(\ref{eq:Capacities-Relation}), we conclusively prove that the
classical private capacity, in fact even the quantum capacity, of
two channels can be greater than the sum of their individual
classical private capacities. Our findings not only demonstrate that
the classical private capacity of a quantum channel is generally not
additive, but also yield another counterexample to the additivity of
quantum capacity, of which the underlying reasoning is
different from that of Smith and Yard's~\cite{SY08}.

\medskip\noindent
{\bf The channel construction.} In the Stinespring representation
$\cN(\rho) =  \tr_E V\rho V^\dagger$, the partial trace embodies all
the noise of the channel as loss of information; if Bob got $E$ as
well, there would be no noise at all as he can undo the isometry.
However, a well-known way of giving him $E$ anyway, is to completely
randomize it: denoting the discrete Weyl operators on $E$ by $W_j$
($j=1,\ldots,|E|^2$), if the channel internally picks $j$ uniformly
at random and applies $W_j$ to $E$, it creates a new channel with
output $\cN(\rho)^B \ox \frac{1}{|E|}\1^E$. The extra register is
always constant, so the new channel has the exact same information
properties as $\cN$. The idea of the following channel construction
is to add another ``gadget'' on top of this, which outputs some
randomness approximating the uniform distribution above -- see
Fig.~\ref{fig1}; so, intuitively, on its own it does not alter too
much the classical capacity of the channel, but if paired with the
right resources can increase the quantum capacity.
\begin{figure}[ht]
  \includegraphics[width=8.5cm]{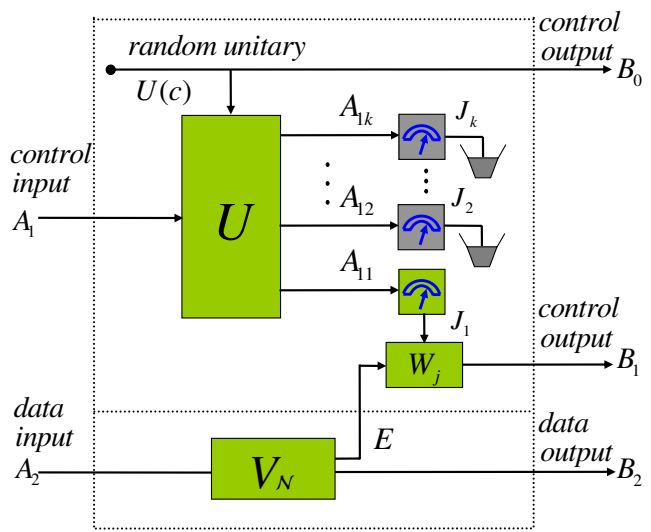}
  \caption{The channel $\cT^k_{\cN}$: in the lower part it contains
    $\cN$ (with input register $A_2$, output register $B_2$, and
    environment $E$).
    It also has another input register $A_1$ of dimension $c=|E|^{2k}$,
    which we view in a fixed way as a tensor product of $k$
    $|E|^2$-dimensional systems $A_{11}$, \ldots, $A_{1k}$, each coming
    with a fixed computational basis $\{ \ket{j} \}_{j=1,\ldots,|E|^2}$. This
    big register is subjected to a random unitary rotation $U$,
    where $U$ is chosen from the Haar measure and subsequently output
    (a classical description of it) in register $B_0$. All registers
    $A_{12}$, \ldots, $A_{1k}$ are discarded, only $A_{11}$ is measured
    in the computational basis, and the result $j$ used to control
    a unitary transformation (Weyl operator) $W_j$ on the environment $E$, which
    is then output in the register $B_1$. A formal definition can be found
    in the Auxiliary Material~\cite{AM}.
    }
  \label{fig1}
\end{figure}
%

A comment on why we need the rather large register
$A_1$, most of which is discarded anyway.
In fact, the size (parametrized by $k$) has a double purpose:
on the one hand, we need $A_{11}$ to be close to maximally mixed
for most inputs. But more importantly, to make it very ``costly'',
though not impossible,
to use entanglement with another system to access the index $J_1$
(see the proof of Theorem~\ref{thm:A1-almost-irrelevant}).

\medskip\noindent
{\bf The additivity violation.}
Now, if we knew that the Holevo quantity $\chi(\cT^k_{\cN})$ were
additive for $\cT^k_{\cN}$, we would have $C(\cT^k_{\cN}) =
\chi(\cT^k_{\cN})$. Since it is possible to show that
$\chi(\cT^k_{\cN}) \leq \chi(\cN) + o(1)$ -- this is a special case
of Theorem~\ref{thm:A1-almost-irrelevant} below --, we would have
an upper bound
\begin{equation}
  \label{eq:P-upper-bound}
  P(\cT^k_{\cN}) \leq C(\cT^k_{\cN}) \leq C(\cN) + o(1).
\end{equation}
While we are not able to show additivity for the channels
$\cT^k_{\cN}$, the above relation is nevertheless true. In fact, we
have the following general theorem, proved in full in the 
Auxiliary Material section~\cite{AM}.
\begin{theorem}
  \label{thm:A1-almost-irrelevant}
  For any channel $\cN$ with input space $A$, output space $B$ and
  envirnoment $E$, and any integer $k$,
  let $\delta(k) = \frac{1}{k}(5+4\log|E|)$.
  Then, for arbitrary channel $\cE$,
  \begin{equation}
    \label{eq:A1-almost-irrelevant}
    \chi(\cN\ox\cE) \leq \chi(\cT^k_{\cN} \ox \cE) \leq \chi(\cN\ox\cE) + \delta(k).
  \end{equation}
  As a consequence,
  \[
    C(\cN) \leq C(\cT^k_{\cN}) \leq C(\cN) + \delta(k).
  \]
\end{theorem}

\medskip
On the other hand, we can get a lower bound on $P(\cT^k_{\cN}
\otimes \cA)$, where $\cA$ is the $50\%$ erasure channel of input
dimension $c$; note that by the no-cloning principle, $P(\cA) =
Q(\cA) = 0$. Since the private classical capacity is not smaller
than the quantum capacity, which is in turn lower bounded by the
coherent information, we evaluate the coherent
information of the channel $\cT^k_{\cN} \otimes \cA$. Let us look at
an input state as follows: Alice prepares a maximally entangled
state $\Phi^{A_1C}$ of dimension $c\times c$ and feeds the two
halves into the control input ($A_1$) and the $50\%$ erasure channel
($C$; its quantum output we also denote $C$, and the erasure flag
$D$). She feeds another arbitrary state $\rho^{A_2}$, whose
purification is denoted as $\ket{\varphi}^{AA_2}$, into the data
input and keeps the system $A$. I.e., we compute the coherent
information with respect to the input state
\(
  \sigma^{A_1A_2C} = \Phi^{A_1C} \otimes \rho^{A_2},
\)
so that the final state after the channel action is
\[
  \omega^{AB_0B_1B_2CD} = (\id_A\ox\cT^k_{\cN}\ox\cA_C)(\Phi^{A_1C} \ox \varphi^{AA_2}).
\]
The coherent information, with respect to this state, is
\[
  I_c(\sigma^{A_1A_2C}, \cT^k_{\cN} \ox \cA) = H(B_1B_2CD|B_0) - H(AB_1B_2CD|B_0).
\]

By an argument similar to that in~\cite{SS08}, we divide the
computation into two cases: the information sent into $\cA$ is
erased or not erased,
\[
  I_c(\sigma^{A_1A_2C}, \cT^k_{\cN} \ox \cA)
      = \frac{1}{2}(I_c^{\text{erased}}+I_c^{\text{not-erased}}).
\]
In the erased case, the state is decoupled between
$AB_2$ and $B_0B_1C$, so the coherent information simplifies to
\[\begin{split}
  I_c^{\text{erased}} &= H(B_2)-H(AB_2) \\
                      &= H\bigl(\cN(\rho)\bigr)
                        -H\bigl( (\id\otimes\cN) \proj{\varphi}^{AA_2} \bigr).
\end{split}\]
When the transmitted information is not erased, Bob
will be able to correct the errors encountered by the noisy channel
$\cN$ as follows. Bob reads the output $B_0$, learning which unitary
transformation is applied by the channel $\cT^k_{\cN}$. Then he can
measure $C$ in the proper basis to get $j$, and then apply
$W_j^\dagger$ to $B_1$, recovering the environment $E$.  As a
result, Bob possesses the output and the environment of $\cN$
simultaneously, effectively obtaining the quantum information input
into $\cN$ completely. In this case, the system $B_0C$ is decoupled
from $AB_1B_2$, which is in the pure state
$(\1\otimes V)\ket{\varphi}^{AA_2}$. So,
\[
  I_c^{\text{not-erased}} = H(B_1B_2)-H(AB_1B_2) = H(\rho^{A_2}).
\]
Adding these two cases together, we have
\[
  I_c(\sigma,\cT^k_{\cN} \otimes \cA)
     = \frac{1}{2} \Bigl[ H(\rho) + H\bigl(\cN(\rho)\bigr)
                     -H\bigl( (\id\ox\cN)\varphi \bigr) \Bigr].
\]
The term in brackets on the right hand side is called quantum mutual
information (between input and output of $\cN$). In~\cite{BSST02},
it is proved that the maximum over $\rho$ of the right hand side is
the entanglement-assisted quantum capacity $Q_E(\cN)$ of the channel
$\cN$. I.e.,
\begin{equation*}
  Q_E(\cN) = \max_\rho I_c(\Phi\otimes\rho,\cT^k_{\cN}\ox \cA),
\end{equation*}
and hence
\begin{equation}
  P(\cT^k_{\cN} \ox \cA) \geq Q(\cT^k_{\cN} \ox \cA) \geq Q_E(\cN).
  \label{transmitting}
\end{equation}

Now, comparing Eqs.~(\ref{eq:P-upper-bound}) and
(\ref{transmitting}), also making use of
Eq.~(\ref{eq:Capacities-Relation}), we see that for all channels
$\cN$ such that
\begin{equation}
  Q_E(\cN) > C(\cN),
  \label{nonadditivity-condition}
\end{equation}
we have, for sufficiently large $k$,
\begin{equation}
  P(\cT^k_{\cN} \ox \cA) \geq Q(\cT^k_{\cN} \ox \cA) > P(\cT^k_{\cN})\geq
  Q(\cT^k_{\cN}).
  \label{additivity-violation}
\end{equation}
Note that the channel $\cA$ has zero private classical capacity and
zero quantum capacity, so Eq.~(\ref{additivity-violation}) exhibits the
violations of the additivity of private classical capacity and the
quantum capacity at the same time.

All we need now is to find
quantum channels that satisfy Eq.~(\ref{nonadditivity-condition}).
One example is the depolarizing
channel of arbitrary dimension $d$,
for which both capacities are known~\cite{King03,BSST02},
$\cD_q(\rho) = (1-q)\rho +q\frac{1}{d}\1$.
For large $d$, the gap becomes asymptotically $\frac{1}{2}H(q,1-q)$~\cite{AM}.

There also exist large additivity violations:
In~\cite[Theorem V.1]{HLSW:rand} it is proven that in sufficiently large
dimension $d$, there exist $n=\bigl\lfloor (\log d)^4 \bigr\rfloor$ orthogonal bases
$\cB_\nu = (\ket{b^{(\nu)}_1},\ldots,\ket{b^{(\nu)}_d})$ such that
for all states $\rho$,
\[
  \frac{1}{n} \sum_{\nu=1}^n H(\cB_\nu|\rho) \geq \log d - 4,
\]
where $H(\cB_\nu|\rho) = -\sum_{i=1}^d \bra{b^{(\nu)}_i}\rho\ket{b^{(\nu)}_i}
                                         \log \bra{b^{(\nu)}_i}\rho\ket{b^{(\nu)}_i}$
is the Shannon entropy of the outcome distribution when
measuring the state $\rho$ in basis $\cB_\nu$. What this means is that
the channel $\cN$ from $d$ to $dn$ dimensions, defined as
\[
  \cN(\rho) = \sum_{\nu=1}^n \sum_{i=1}^d
               \frac{1}{n} \bra{b^{(\nu)}_i}\rho\ket{b^{(\nu)}_i}\ \proj{i}^B \ox \proj{\nu}^{B'},
\]
satisfies $\chi(\cN) \leq 4$. Since the channel is entanglement-breaking,
the additivity result of~\cite{Shor01}
applies, so $C(\cN) = \chi(\cN) \leq 4$. On the other hand, it is
straightforward to see that $Q_E(\cN) = \frac{1}{2}\log d$.
Thus we find that almost the entire bandwidth of $\cN$ can be activated by
the presence of entanglement.
Now, to construct the example for activation of the secret capacity by
a $50\%$ erasure channel, we observe that $|E| = dn = d(\log d)^4$.
We choose $k=\log |E|$, and get from Theorem~\ref{thm:A1-almost-irrelevant}
that $C(\cT_{\cN}^k) \leq O(1)$, while at the same time
$Q(\cT_{\cN}^k \ox \cA) \geq \frac{1}{2}\log d$.
Note however that the
total input dimension is $2^{O\left((\log d)^2\right)}$, which is also
the input dimension of the $50\%$ erasure channel.

\medskip\noindent
{\bf Conclusion.}
We showed a way of converting any gap between classical capacity
and entanglement-assisted quantum capacity of a channel into
a violation of the additivity of the private capacity of
the channel tensored with a $50\%$ erasure channel. In fact,
the \emph{quantum} capacity of the tensor product channel
is larger than the \emph{classical} capacity of the single channel.

The construction is based on a certain embedding of the given
channel into a version of the echo-correctable channels
from~\cite{BDSS04}. That the pairing with the erasure channel gives
larger quantum capacity follows from the echo-correctable reasoning
of the benefit of sharing entanglement. On the other hand, the upper
bound on the classical capacity relies on
showing that the additional ``gadgets'' built around the given
channel increase the capacity by an arbitrarily small amount.
The argument to do so is different from the one proving additivity
of $\chi$ of the channel (which we cannot do for $\cT_{\cN}^k$),
and also from the use of the recent continuity bound~\cite{SmithLeung:cont}
(which cannot be applied as $\cT_{\cN}^k$ is at finite distance
from any channel for which we know the capacity).

Thus, we even get a new type of example for the non-additivity of
the quantum capacity $Q$, which is different from Smith and
Yard's~\cite{SY08} as our channel is not PPT entanglement binding.
Furthermore, while in~\cite{SY08} the lower bound of half the
private capacity on the quantum capacity of the tensor product was
enough, here we experience even a large gap between these two
quantities. However, we also note a conceptual analogy in the
constructions: The PPT entanglement binding channel used
in~\cite{SY08} derives from a so-called \emph{pbit
state}~\cite{HHHO05}. It provides Alice and Bob with shared
randomness -- which is made private by distributing the purification
among Alice and Bob, but in a scrambled way that makes it impossible
for them to recover much of the entanglement. Our channel randomizes
the environment and hence gives it to Bob in an encrypted way,
limiting the receiver's knowledge about the noise
encountered by the channel. In the construction of~\cite{SY08} as in
the present one, the availability of additional resources allows
Alice and Bob to break the encryption and access the entanglement.



\medskip\noindent
{\bf Acknowledgments.} KL thanks the University of Bristol and CQT
at NUS, where part of this work was done, for their hospitality. KL,
XBZ and GCG were supported by the National Fundamental Research
Program (Grant No.~2009CB929601), the National Natural Science
Foundation of China (Grant Nos.~10674128 and 60121503), and the
Innovation Funds and ``Hundreds of Talents'' program of the Chinese
Academy of Sciences and Doctor Foundation of Education Ministry of
China (Grant No.~20060358043). AW was supported by the U.K.~EPSRC,
by the Royal Society, a Philip Leverhulme Prize, and by the European
Commission. The Centre for Quantum Technologies is funded by the
Singapore Ministry of Education and the National Research Foundation
as part of the Research Centres of Excellence programme.

\vfill\pagebreak

\appendix

\section*{The channel construction\protect\\ (Auxiliary material)}
Mathematically, the channel depicted in
Fig.~1 is written
\begin{equation}\begin{split}
  \label{eq:T-definition}
  \cT^k_{\cN}(\rho^{A_1 A_2})
       \!=\! \int_U \!\!\!{\rm d}U [U]^{B_0}
                         &\ox \sum_{j }( W_j^{E\rar B_1})(V^{A_2\rar B_2E})  \\
                         &\!\!\!\!\!\!\!\!\!\!\!\!\!\!\!\!\!\!\!\!\!\!\!\!\!\!\!\!\!\!\!\!\!\!\!\!\!\!\!\!\!\!\!\!\!
                          \bigl(\tr_{A_{1}} \bigl[U^{A_1}\rho {U^{A_1}}^\dagger\proj{j}^{A_{11}}  \bigr]\bigr)
                               ( V^{A_2\rar B_2E})^\dagger( W_j^{E\rar B_1})^\dagger,
\end{split}\end{equation}
where the notation $[U]$ denotes a classical label realized as mutually
orthogonal states. If there are only countably many values of $U$, these
may be thought of as orthogonal projectors $\proj{U}$ on an appropriate
Hilbert space.

\section*{Proof of Theorem 1\protect\\ (Auxiliary material)}
The left inequality in Eq.~(8) is trivial
(simply ignore the registers $A_1$, $B_0$ and $B_1$); also, once
the upper bound is proved, the inequality for the capacity follows
by induction. Hence we concentrate on the right inequality in (8).
We shall need a number of auxiliary results.

\begin{lemma}
  \label{lemma:approximate-j-decoupling}
  Consider an arbitray state $\rho$ on $\CC^r \ox \CC^s$, and a
  fixed basis $\{ \ket{j} \}_{j=1,\ldots,r}$ of $\CC^r$.
  Then for a Haar-distributed random unitary $U \in \mathcal{U}(\CC^r \ox \CC^s)$,
  the random variable $J$ is defined as follows:
  \[
    \Pr\{ J=j | U \} = \bra{j} \tr_s U\rho U^\dagger \ket{j},
  \]
  and it holds that
  \[
    \EE_U H(J|U) \geq \log r - \log\left( 1+\frac{1}{s} \right)
           \geq \log r - \frac{2}{s}.
  \]
\end{lemma}
\begin{proof}
Without loss of generality, the input state is some fixed pure state,
so that after the unitary and before the measurement, we have a
uniformly distributed random state $\ket{\phi}$.
Using the bound $H(J) = -\sum_j P_J(j)\log P_J(j) \geq -\log \sum_j P_J(j)^2$
and the concavity of $\log x$, we get now by symmetry
\begin{equation}\begin{split}
  \label{eq:cond-J-entropy}
  \EE_U H(J|U) &\geq -\log \EE_\phi \sum_{j=1}^r \bra{j} \tr_s \phi \ket{j}^2 \\
               &=    -\log\left( r \EE_\phi \bra{1} \tr_s \phi \ket{1}^2 \right) \\
               &=    -\log\left( r \EE_\phi \bigl[ \tr\phi(\proj{1}^r\ox\1^s) \bigr]^2 \right).
\end{split}\end{equation}
For the latter expectation, we use a well-known trick:
\begin{equation*}\begin{split}
  \EE_\phi &\bigl[ \tr\phi(\proj{1}^r\ox\1^s) \bigr]^2  \\
           &= \EE_\phi \tr\bigl( (\phi\ox\phi)(\proj{1}^r\ox\1^s\ox\proj{1}^r\ox\1^s) \bigr) \\
           &= \tr\left( \frac{\1+F}{rs(rs+1)}(\proj{1}^r\ox\1^s\ox\proj{1}^r\ox\1^s) \right),
\end{split}\end{equation*}
where we have introduced a second tensor copy of the total Hilbert space
$\CC^r\ox\CC^s$, and $F$ is the SWAP operator of the two. The last line
evaluates easily to
\begin{equation*}\begin{split}
  \EE_\phi \bigl[ \tr\phi(\proj{1}^r\ox\1^s) \bigr]^2
           &=    \frac{1}{rs(rs+1)}(s^2+s)           \\
           &=    \frac{1}{r^2}\frac{s+1}{s+1/r}
            \leq \frac{1}{r^2}\frac{s+1}{s}.
\end{split}\end{equation*}
Inserting this into Eq.~(\ref{eq:cond-J-entropy}) concludes the proof.
\end{proof}

\begin{remark}
  \label{rem:finite-U}
  From the above calculation we see that the full unitary invariance of
  the Haar measure is not required; we only need to be able to perform the
  average purity of the reduced state, which is a quadratic function
  of the random pure state $\varphi$. Thus, it is sufficient to draw $U$
  from a so-called \emph{unitary 2-design}~\cite{Dankert06}; in
  prime power dimension it is known that the \emph{Clifford group}
  is an example of a finite 2-design, see e.g.~\cite{Gross07}.
\end{remark}

Now, in the channel we imagine that, just as $A_{11}$, also the
$k-1$ registers $A_{12},\ldots,A_{1k}$ are measured in a fixed
basis, resulting in $k$ random variables $J_1,\ldots,J_k$. For
simplicity, sometimes we write $J_\ell\ldots J_k$ as $J_\ell^k$, and
similarly for the measurement results, $j_\ell$. The variables
$J_2,\ldots,J_k$ are being traced over, so the channel does not
change (See Fig.~1).

\begin{lemma}
  \label{n-to-one}
  For any state $\sigma^{A_1A_2B_3}$, suppose we feed $A_1$
  into the channel $\cT^k_{\cN}$ but keep $A_2$ unchanged, let $\omega^{B_0J_1^k A_2B_3}$
  be the state after applying the random unitary $U$ and then doing
  measurements on $A_{1\ell}$, i.e.
  \[\begin{split}
    \omega^{B_0J_1^k A_2B_3} &= \int_U {\rm d}U [U]^{B_0} \!
                                       \ox \!\!\!\sum_{j_1\ldots j_k} \bigotimes_{\ell=1}^k [j_\ell]^{J_\ell}
                                        \!                                                         \\
                     &\phantom{==}
                      \ox\!
                      \bra{j_1\ldots j_k} U^{A_1} \sigma^{A_1A_2B_3} {U^{A_1}}^\dagger
                                                      \ket{j_1\ldots j_k}.
  \end{split}\]
  Then we have
  \begin{equation}\begin{split}
    I(J_1;A_2B_3|B_0) &\leq \frac{1}{k}I(J_1^k;A_2B_3|B_0) \\
                       &\phantom{==}
                            +\frac{1}{k} \sum_{\ell=1}^{k-1} I(J_\ell;J_{\ell+1}^k|B_0),
  \label{ineq0}
  \end{split}\end{equation}
  where
  \[\begin{split}
    I(X:Y) &= H(X)+H(Y)-H(XY) \\
           &= H(\rho_X)+H(\rho_Y)-H(\rho_{XY})
  \end{split}\]
  is the (quantum) mutual information
  between two subsystems of a bipartite
  state $\rho_{XY}$ with marginals $\rho_X$ and $\rho_Y$,
  and the informations conditional on $B_0$ are averages over
  the classical states of this register.
\end{lemma}
\begin{proof}
We use the chain rule to get
\[\begin{split}
  I(J_1^k;A_2B_3|B_0)
     &=    I(J_1;A_2B_3|B_0) + I(J_2^k;A_2B_3|B_0J_1) \\
     &=    I(J_1;A_2B_3|B_0) + I(J_2^k;A_2B_3J_1|B_0) \\
     &\phantom{===}                - I(J_1;J_2^k|B_0) \\
     &\geq I(J_2^k;A_2B_3|B_0) \\
     &\phantom{==}
           + I(J_1;A_2B_3|B_0) - I(J_1;J_2^k|B_0).
\end{split}\]
Iterating this step for $I(J_2^k;A_2B_3|B_0)$, then
$I(J_3^k;A_2B_3|B_0)$, etc., we obtain
\begin{equation}\begin{split}
  \label{ineq1}
  I(J_1^k;A_2B_3|B_0) &\geq \sum_{\ell=1}^k I(J_\ell;A_2B_3|B_0) \\
                       &\phantom{==}
                             - \sum_{\ell=1}^{k-1} I(J_\ell;J_{\ell+1}^k|B_0).
\end{split}\end{equation}
By symmetry all the $I(J_\ell;A_2B_3|B_0)$ are the same and equal to
$I(J_1;A_2B_3|B_0)$, concluding the proof.
\end{proof}

Now, to show Eq.~(8), we need to do the
following. Given another channel $\mathcal{E}$, whose input and
output registers we denote $A_3$ and $B_3$, respectively, we first
have to show that correlated inputs between $A_1$ and $A_2A_3$ are
(almost) of no use, and hence that the control part of our channel
$\cT^k_{\cN}$ can be skipped.

Mathematically, we have to look at a given ensemble of input states
$\{\lambda_x, \varphi_x^{A_1A_2A_3}\}$ for $\cT^k_{\cN} \ox \cE$. We
shall find a new ensemble of states only on $A_2A_3$ which has
almost the same Holevo information, even when we consider only the
output registers $B_2B_3$, i.e.~those of $\cN\ox\cE$. It turns out
that we have to distinguish two cases: An individual state
$\varphi^{A_1A_2A_3}$ can result in ``small'' correlation between
$J_1^k$ and $A_2B_3$ -- but then Lemma~\ref{n-to-one} above limits
the correlation between $J_1$ and $A_2B_3$, making input $A_1$ of
almost no use (Proposition~\ref{less-four} below). On the other
hand, if there is ``large'' correlation, we can use the $J_\ell$ to
break up the input state into an ensemble acting only on $A_2A_3$
with at least the same contribution to the Holevo quantity
(Proposition~\ref{geq-four}). For the following two Propositions, we
do the same thing as above, keeping a record of $J_1^k$, together
with the output state on $B_0B_2B_3$. However, notice that after a
unitary $U$ is applied to $A_1$ and then the measurements on the
$A_{1\ell}$ performed, $A_2A_3$ collapses into a state
$\varphi_{j_1^k}(U)^{A_2A_3}$ with probability $p_{j_1^k}(U)$:
\[
  p_{j_1^k}(U)\varphi_{j_1^k}(U)^{A_2A_3}
    = \bra{j_1\ldots j_k} U^{A_1} \varphi {U^{A_1}}^\dagger \ket{j_1\ldots j_k}.
\]
With respect to which \emph{input} state entropic quantities
are to be interpreted is indicated by adding that input state
as a subscript (unless it is $\varphi$).

\begin{proposition}
  \label{less-four}
  For the joint channel $\cT^k_{\cN} \otimes \cE$ with input
  state $\varphi^{A_1A_2A_3}$ and
  $\sigma^{A_1A_2B_3} = (\id_{A_1A_2}\ox\cE)\varphi^{A_1A_2A_3}$, if
  $I(J_1^k;A_2B_3|B_0) < 4\log|E|$, then
  \[
    H(B_1B_2B_3|B_0) > H(B_2B_3) + \log|E| - \delta(k),
  \]
  where $\delta(k) = \frac{1}{k}(5+4\log|E|)$.
\end{proposition}
\begin{proof}
We start by invoking Lemma~\ref{lemma:approximate-j-decoupling}
to the entropy of $J_\ell\ldots J_k$: this yields
\[
  H(J_\ell\ldots J_k|B_0) \geq 2(k-\ell+1) \log |E| - \frac{2}{|E|^{2(\ell-1)}},
\]
which implies $I(J_\ell;J_{\ell+1}^k|B_0) \leq
\frac{2}{|E|^{2(\ell-1)}}$.

Now, with Lemma~\ref{n-to-one} and using the assumption that
$I(J_1^k;A_2B_3|B_0) < 4\log |E|$, we find
\begin{equation}\begin{split}
\label{mutual-info-bound}
  I(J_1;A_2B_3|B_0) &\leq \frac{1}{k}I(J_1^k;A_2B_3|B_0)  \\
                     &\phantom{\leq \frac{1}{k}4\log |E|.}
                        + \frac{1}{k}\sum_{\ell=1}^{k-1} I(J_\ell;J_{\ell+1}^k|B_0) \\
                     &\leq \frac{1}{k}4\log |E| + \frac{1}{k}\sum_{\ell=1}^\infty \frac{2}{|E|^{2(\ell-1)}} \\
                     &=    \frac{1}{k}4\log |E| + \frac{1}{k}\frac{2}{1-1/|E|^2} \\
                     &\leq \frac{1}{k}(3+4\log|E|).
\end{split}\end{equation}
For the given input state $\varphi^{A_1A_2A_3}$, let us write the
quantum state of the system $J_1A_2B_3$ as $\omega^{J_1A_2B_3}(U)$,
and the output state on $B_1B_2B_3$ as $\phi^{B_1B_2B_3}(U)$. We
also denote the CPTP quantum operation mapping $J_1A_2B_3$ to
$B_1B_2B_3$ in our setting of $\cT^k_{\cN} \ox \cE$ as $\cR$. Then
\[\begin{split}
  \cR\Bigl(\omega^{J_1A_2B_3}(U)\Bigr) &=\phi^{B_1B_2B_3}(U),   \\
  \cR\left(\frac{1}{|E|^2}\1^{J_1}\otimes\omega^{A_2B_3}\right)
                                       &=\frac{1}{|E|}\1^{B_1}\otimes\phi^{B_2B_3}.
\end{split}\]
By straightforward calculation, and using Eq.~(\ref{mutual-info-bound}) and
Lemma~\ref{lemma:approximate-j-decoupling} once more, we have
\[\begin{split}
  \EE_U &D\Bigl(\omega^{J_1A_2B_3}(U)\Big\|\frac{1}{|E|^2}\1^{J_1}\otimes\omega^{A_2B_3}\Bigr) \\
        &=    2\log|E|+I(J_1;A_2B_3|B_0)-H(J_1|B_0) \\
        &\leq 2\log|E|+\frac{1}{k}(3+4\log|E|)-\left( 2\log|E|-\frac{2}{|E|^{2(k-1)}} \right) \\
        &\leq \frac{1}{k}(5+4\log|E|),
\end{split}\]
where $D(\rho\|\sigma):=\tr\bigl(\rho(\log\rho-\log\sigma)\bigr)$ is the
quantum relative entropy. By the Lindblad-Uhlmann theorem~\cite{Uhlmann77}, the
quantum relative entropy is monotonic under completely positive
quantum operations. As a result, we obtain
\[\begin{split}
  \log|E|- &H(B_1B_2B_3|B_0)+H(B_2B_3)\\
           &=    \EE_UD\Bigl(\phi^{B_1B_2B_3}(U)\Big\|\frac{1}{|E|}\1^{B_1}\otimes\phi^{B_2B_3}\Bigr)\\
           &\leq \EE_UD\Bigl(\omega^{J_1A_2B_3}(U)\Big\|\frac{1}{|E|^2}\1^{J_1}\otimes\omega^{A_2B_3}\Bigr)\\
           &\leq \frac{1}{k}(5+4\log|E|),
\end{split}\]
which is exactly to say
\[\begin{split}
  H(B_1B_2B_3|B_0)\geq H(B_2B_3) + \log |E| - \frac{1}{k}(5+4\log|E|),
\end{split}\]
and we are done.
\end{proof}

\begin{proposition}
  \label{geq-four}
  For the joint channel $\cT^k_{\cN} \otimes \cE$ with input
  state $\varphi^{A_1A_2A_3}$ and
  $\sigma^{A_1A_2B_3} = (\id_{A_1A_2}\ox\cE)\varphi^{A_1A_2A_3}$, if
  $I(J_1^k;A_2B_3|B_0) \geq 4\log|E|$, then
  there exists a particular unitary $U_0$ such that
  \[
     H(B_1B_2B_3|B_0) \geq \sum_{j_1\ldots j_k} p_{j_1^k}(U_0) H_{\varphi_{j_1^k}(U_0)}(B_2B_3) + \log|E|.
  \]
\end{proposition}
\begin{proof}
By assumption, we have
\[\begin{split}
  4\log|E| &\leq I_{\varphi^{A_1A_2A_3}}(J_1^k;A_2B_3|B_0) \\
           &=    I_{\varphi^{A_1A_2A_3}}(J_1^k;EB_2B_3|B_0) \\
           &\leq I_{\varphi^{A_1A_2A_3}}(J_1^k;B_2B_3|B_0) +2\log|E|,
\end{split}\]
since the discarding of the register $E$ cannot reduce the mutual
information by more than $2\log|E|$. Thus,
\begin{equation}\begin{split}
   2\log|E| &\leq I_{\varphi^{A_1A_2A_3}}(J_1^k;B_2B_3|B_0) \\
           &=    H_{\varphi^{A_2A_3}}(B_2B_3) - H_{\varphi^{A_1A_2A_3}}(B_2B_3|J_1^k B_0).
  \label{num1}
\end{split}\end{equation}
Hence, there must be a special unitary $U_0$, such that
\begin{equation}\begin{split}
  H_{\varphi^{A_1A_2A_3}}(B_2B_3|J_1^k B_0)
      &\geq H_{\varphi^{A_1A_2A_3}}(B_2B_3|J_1^k,B_0=U_0)   \\
      &\!\!\!\!\!
       =    \sum_{j_1^k}  p_{j_1^k}(U_0) H_{\varphi_{j_1^k}(U_0)^{A_2A_3}}(B_2B_3).
  \label{num2}
\end{split}\end{equation}
By the subadditivity of von Neumann entropy, we have
\begin{equation}\begin{split}
  \label{num3}
  H_{\varphi^{A_1A_2A_3}}(B_1B_2B_3|B_0)
          &\geq H_{\varphi^{A_2A_3}}(B_2B_3)-H(B_1|B_0)\\
          &\geq H_{\varphi^{A_2A_3}}(B_2B_3)-\log|E|.
\end{split}\end{equation}
At last, putting Eqs.~(\ref{num1}-\ref{num3})
together completes the proof.
\end{proof}

With that we are now ready for the

\medskip
\begin{proof}[of Theorem 1]
For an ensemble of input
states $\{\lambda_x, \varphi_x^{A_1A_2A_3}\}$ for $\cT^k_{\cN} \ox
\cE$ with average state $\rho = \sum_x \lambda_x\varphi_x$, we
divide it up into two classes according to the above cases:
\begin{align*}
  \mathcal{G} &:=\{ x | I_{\varphi_x}(J_1^k;A_2B_3|B_0) \geq 4\log|E| \}, \\
  \mathcal{L} &:=\{ x | I_{\varphi_x}(J_1^k;A_2B_3|B_0) <    4\log|E| \}.
\end{align*}
Then,
\begin{equation}\begin{split}
  \chi_{\{\lambda_x,\varphi_x\}}(\cT^k_{\cN} \ox \cE)
      &= H_{\rho}(B_1B_2B_3|B_0) \\
      &\phantom{=} - \sum_{x\in\cG} \lambda_x H_{\varphi_x}(B_1B_2B_3|B_0) \\
      &\phantom{=} -\sum_{x\in\cL} \lambda_x H_{\varphi_x}(B_1B_2B_3|B_0).
  \label{theorem-1}
\end{split}\end{equation}
By the subadditivity of the von Neumann entropy, the first term is
upper bounded
\begin{equation}\begin{split}
  H_\rho(B_1B_2B_3|B_0) &\leq H_\rho(B_1|B_0) + H_\rho(B_2B_3|B_0) \\
                        &\leq H_{\rho^{A_2A_3}}(B_2B_3) + \log|E|.
  \label{theorem-2}
\end{split}\end{equation}
Second, by Proposition~\ref{less-four}, we have
for $x\in\mathcal{L}$ that
\begin{equation}
  H_{\varphi_x}(B_1B_2B_3|B_0)
      > H_{\varphi_x^{A_2A_3}}(B_2B_3) + \log|E| - \delta(k).
  \label{theorem-3}
\end{equation}
Third, by Proposition~\ref{geq-four}, for $x\in\mathcal{G}$, there is
an ensemble decomposition of $\varphi^{A_2A_3}_x$,
\[
  \varphi^{A_2A_3}_x =\sum_{y=1}^{|E|^{2k}} \mu_{xy} \varphi^{A_2A_3}_{xy},
\]
such that
\begin{equation}
  H_{\varphi_x}(B_1B_2B_3|B_0)
       \geq \sum_{y=1}^{|E|^{2k}} \mu_{xy} H_{\varphi_{xy}^{A_2A_3}}(B_2B_3) + \log|E|.
  \label{theorem-4}
\end{equation}
Now define the union ensemble of the above states,
\[
  \cO = \{ \lambda_x,\varphi_x^{A_2A_3} \}_{x\in\cL}
           \cup \{ \lambda_x \mu_{xy}, \varphi_{xy}^{A_2A_3} \}_{x\in\cG,y=1,\ldots,|E|^{2k}}.
\]
Inserting Eqs.~(\ref{theorem-2}-\ref{theorem-4}) into Eq.~(\ref{theorem-1})
results in
\[
  \chi_{\{\lambda_x,\varphi_x\}} (\cT^k_{\cN} \ox \cE)
      \leq \chi_{\mathcal{O}}(\cN \otimes \cE) + \delta(k),
\]
and we are done.
\end{proof}

\section*{Depolarizing channel\protect\\ (Auxiliary material)}
The difference of $Q_E(\mathcal{D}_{q})$ and
$C(\mathcal{D}_{q})$ of the depolarizing channel
$\mathcal{D}_{q}$, for several special values of $d$.
The following plot shows $q$ (horizontal axis)
against $Q_E(\mathcal{D}_{q})-C(\mathcal{D}_{q})$
on the vertical axis.
\setcounter{figure}{1}
\begin{figure}[ht]
  \begin{center}
  \includegraphics[width=7cm]{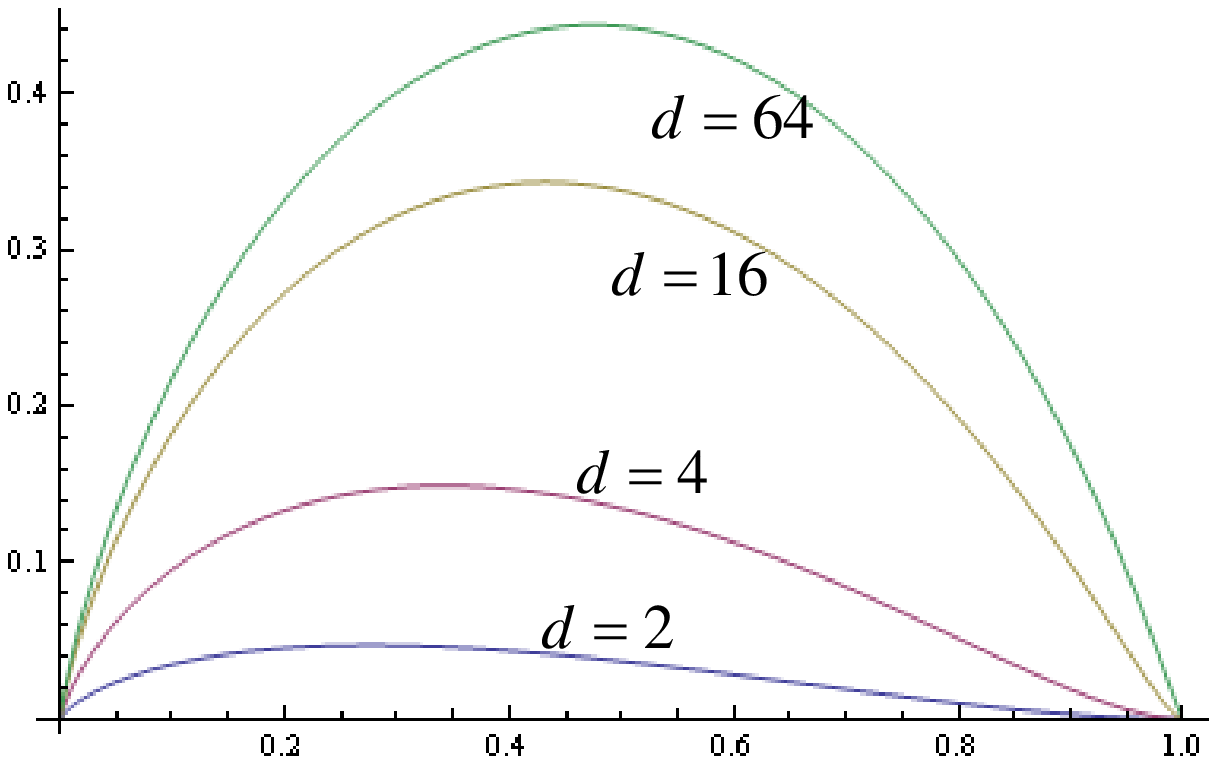}
  \end{center}
  \label{fig:graph}
\end{figure}

\end{document}